\documentclass[runningheads]{llncs}
\usepackage{amsmath} 
\usepackage[pdftex]{graphicx}
\usepackage{tabularx}
\usepackage{graphicx}
\usepackage{subfig}
\usepackage{caption}
\usepackage{bbm}
\usepackage{amsfonts}
\usepackage{enumitem}

\usepackage{algorithm}
\usepackage{algpseudocode}
\usepackage{mathrsfs}

\usepackage{xcolor}
\begin{document}
\title{On the Characterization of Saddle Point Equilibrium for Security Games with Additive Utility\thanks{This work was supported in part by NSF CPS grant ECCS 1739969.}}
\titlerunning{On the Characterization of Saddle Point Equilibrium for Security Games}
%
\author{Hamid Emadi\inst{1}\orcidID{0000-0002-6707-9679} \and
Sourabh Bhattacharya\inst{1}\orcidID{0000-0002-1306-2775} }
\authorrunning{H. Emadi and S. Bhattacharya}
%
\institute{Iowa State University,  
   Ames, IA 50011 USA \\
\email{\{emadi,sbhattac\}@iastate.edu}}
\maketitle              
\begin{abstract}
In this work, we investigate a security game between an attacker and a defender, originally proposed in \cite{emadi2019security}. As is well known, the combinatorial nature of security games leads to a large cost matrix. Therefore, computing the value and optimal strategy for the players becomes computationally expensive. In this work, we analyze a special class of zero-sum games in which the payoff matrix has a special structure which results from the {\it additive property} of the utility function. Based on variational principles, we present structural properties of optimal attacker as well as defender's strategy. We propose a linear-time algorithm to compute the value based on the structural properties, which is an improvement from our previous result in \cite{emadi2019security}, especially in the context of large-scale zero-sum games.

\keywords{Security Game, Zero-Sum Game,  Nash Equilibrium and Computational Complexity.}
\end{abstract}
\section{Introduction}
Game theory \cite{basar1999dynamic} is a useful tool to model adversarial scenarios. \textit{Security games} model attack scenarios wherein an attacker attacks a number of targets while the defender allocates its resources to protect them to minimize the impact. The payoff for the attacker and the defender is based on the successfully attacked and protected targets, respectively. Traditionally, attacker-defender games have been modeled as zero-sum games, and the resulting saddle-point strategies are assumed to be optimal for both players. In general, two-player zero-sum game can be formulated as an linear programming problem~\cite{basar1999dynamic}, and therefore saddle-point equilibrium can be computed in polynomial time. The most efficient running time of solver for a general LP problem is ${\cal{O}}(n^{2.055})$~\cite{jiang2020faster}. However, solving security games with more than 2 resources for attacker and defender with a general LP solver is computationally expensive due to the combinatorial nature of the problem. 



In the past two decades, game theory has played an important role in quantifying and anlyzing security in large-scale networked systems. Here, we mention some of these efforts across several applications. In~\cite{boudko2018evolutionary}, authors propose an evolutionary game framework that models integrity attacks and defenses in an Advanced Metering Infrastructure in a smart grid. In~\cite{lim2012game}, a game-theoretic defense strategy is developed to protect sensor nodes from attacks and to guarantee a high level of trustworthiness for sensed data. In~\cite{law2014security}, authors consider a security game in power grid in which the utility functions are defined based on the defender's net loss, and players' action set are defined based on different scenarios such as false data injection to electronic monitoring systems. In~\cite{shan2020game}, authors provide a game theoretic approach to modeling attack and defense of smart grid in different layers of power plant, transmission and distribution network. In~\cite{hamdi2014game},\cite{boudko2019adaptive}, authors propose a game-theoretic model for the adaptive security policy and the power consumption in the Internet of Things. In~\cite{pirani2019game}, authors propose a zero-sum game for security-aware sensor placement in which the attacker minimizes its visibility by attacking certain number of nodes, and the detector maximizes the visibility of the attack signals by placement of sensors. In~\cite{nugraha2019subgame}, authors consider a cyber-security problem in a networked system as a resilient graph problem. They use backward induction to obtain optimal strategies for both players which try to disrupt the communication and maintain the connectivity, respectively. In~\cite{nugraha2019subgame}, authors propose a stochastic game-based model to provide a control synthesis method to minimize deviations from the desired specification due to adversaries, who has limited capability of observing the controller's strategy. In~\cite{grossklags2008secure}, authors focus on notions of self-protection (e.g., patching system vulnerabilities) and self insurance (e.g., having good backups) rather than only security investments in information security games. ~\cite{johnson2010uncertainty} introduces a method for resolving uncertainty in interdependent security scenarios in computer network and information security. In~\cite{trajanovski2015designing}, authors examine security game in which each player collectively minimize the cost of virus spread, and assuring connectivity. In~\cite{khouzani2015picking}, authors propose a game theoretic framework for picking vs guessing attacks in the presence of preferences over the secret space, and they analyse the trade-off between usability and security. In~\cite{laszka2018game}, authors introduce a game-theoretic framework for optimal stochastic message authentication, and they provide guarantees for resource-bounded systems. For a comprehensive survey of game-theoretic approaches in security and privacy in computer and communication, a reader could refer to \cite{manshaei2013game}.

Security games pose computational challenges in analysis and synthesis of optimal strategies due to exponential increase in the size of the strategy set for each player. A class of security games which renders tractable computational analysis is that of Stackelberg games \cite{basar1973relative}. In Stackelberg models, the leader moves first, and the follower observes the leader’s strategy before acting. Efforts involving randomized strategies~\cite{kiekintveld2009computing} and approximation algorithms \cite{bhattacharya2011approximation} for Stackelberg game formulation of security games have been proposed to efficiently allocate multiple resources across multiple assets/targets. In order to extend the efficient computational techniques for simultaneous move games, efforts have been made to characterize conditions under which any Stackelberg strategy is also a NE strategy~\cite{korzhyk2011stackelberg}. An extensive review of various efforts to characterize and reduce the computational complexity of Stackelberg games with application in security can be found in~\cite{emadi2019security}, and references therein. Here, we mention a few that are relevant to the problem under consideration. \cite{korzhyk2010complexity} shows that computing the optimal Stackelberg strategy in security resource allocation game, when attacker attacks one target, is NP-hard in general. However, when resources are homogeneous and cardinality of protection set is at most 2, polynomial-time algorithms have been proposed by the authors. \cite{korzhyk2010complexity} propose an LP formulation similar to Kiekintveld's formulation, and presents a technique to compute the mixed strategies in polynomial time.

In~\cite{korzhyk2011security}, a security game between an attacker and a defender is modeled as a non-zero game with multiple attacker resources. The authors analyze the scenario in which the payoff matrix has an {\it additive structure}. They propose an $\mathcal{O}(m^2)$ iterative algorithm for computing the mixed-strategy Nash Equilibrium where $m$ is the size of the parameter set. Motivated from \cite{korzhyk2011security}, in~\cite{emadi2019security}, we analyzed a zero-sum security game with multiple resources for attacker and defender in which the payoff matrix has an additive structure. Based on combinatorial arguments, we presented structural properties of the saddle-point strategy of the attacker, and proposed an $\mathcal{O}(m^2)$ algorithm to compute the saddle-point equilibrium and the value of the game, and provided closed-form expressions for both. In this paper, we show that a zero-sum security game can be reduced to the problem of minimizing the sum of the $k$-largest functions over a polyhedral set which can be computed in linear time~\cite{ogryczak2003minimizing}. Based on this insight, we use a variational approach to propose an $\mathcal{O}(m)$ algorithm which is the best possible in terms of the complexity.
Moreover, we present structural properties of the saddle-point strategy of both players, and an explicit expression for the value of the game.

The rest of the paper is organized as follows. In Section 2, we present the problem formulation. In Section 3, we present structural properties of the optimal attacker strategy. In Section 4, we present a linear time algorithm to compute the value of a large-scale zero-sum game. In Section 5, we present structural properties of the defender's optimal strategy, and a dual algorithm to compute the value and equilibrium. In section 6, we present our conclusions along with some future work.

\section{Problem Formulation: Security Game}

Consider a two-person zero-sum game, and let ${\cal{I}}=\left\lbrace 1,\dots m\right\rbrace $ denotes a set of targets. We assume an attacker (player 1) chooses $k_a$-targets to attack. So, there are $n_a={{m}\choose{k_a}}$ actions for player 1. On the other hand, protection budget of targets is limited, and we assume that only $k_d$ targets will be protected by the defender (player 2). So, there are $n_d={{m}\choose{k_d}}$ actions for player 2.  The defender has no knowledge about the targets chosen by player 1.  In order to find the optimal strategy for the players, we formulate a strategic security game $({\cal{X}},{\cal{Y}},A)$, where $\cal{X}$ and $\cal{Y}$ denote the action sets for attacker and defender, respectively, and $\text{card}({\cal{X}})=n_a,\text{card}({\cal{Y}})=n_d$. Every element $x_i\in {\cal{X}}$ represents a set of targets that are attacked. Similarly, ${y}_i \in {\cal{Y}}$ represents a protected targets. Each ${x}_i\in {\cal{X}}$ and ${y}_i \in {\cal{Y}}$ is a $k_a$-tuple, and $k_d$-tuple subset of ${\cal{I}}$, respectively.


The attacker has no information about the targets that are protected by the defender. Let $\phi_i$ denote the cost associated to target $i$. Moreover, without loss of generality, we assume that targets are labeled such that $\phi_i\ge \phi_j\ge 0$ for $i>j$. 

We consider an additive property for the utility function i.e., entries of the cost matrix $A$ are defined as follows:
\begin{eqnarray}\label{game matrix elements}
A_{ij} = \sum_{\left\lbrace l|l\in x_i \cap y_j^c \right\rbrace }^{}{\phi_l}.
\end{eqnarray}
$A$ represents the game matrix or payoff matrix for player 1. Since we consider a zero-sum game, the payoff matrix for player 2 is $-A$. Note that we assume both players have the complete information of the target costs.  

Let $p,q$ be the probability vectors representing the mixed strategies for player 1 and player 2, respectively. The expected utility function is 
\begin{eqnarray*}
v=p^TAq.
\end{eqnarray*} 
According to the minimax theorem, every finite two-person zero-sum game has a saddle point with the value, $v^*$, in mixed strategy $p^*=\left[{p^*_1},\dots,{p^*_{n_a}}\right]^T $  for player 1, and mixed strategy $q^*=\left[{q^*_1},\dots,{q^*_{n_d}}\right]^T$  for player 2, such that player 1's average gain is at least $v^*$ no matter what player 2 does. And player 2's average loss is at most $v^*$ regardless of player 1's strategy, that is
\begin{eqnarray}\nonumber
{p}^TAq^* \le {p^*}^TAq^*\le {p^*}^TAq.
\end{eqnarray}
 In order to solve every finite matrix game, we can reduce the game to the following LP problem,  
 \begin{equation} \label{original LP}
 	\begin{aligned}
 		&\underset{p_i}{\text{maximize}}
 		& &  v \\
 		& \text{subject to}
 		& & v\leq \sum_{i=1}^{n_a}p_ia_{ij},\quad j=1,\dots,n_d\\
 		&&& p_1+\dots+p_{n_a}=1\\
 		&&& p_i \geq 0\quad \text{for}\quad i=1,\dots,n_a.
 	\end{aligned}
 \end{equation}
However, the dimension of the decision variables in the above formulation is $(n_a+1)$ which is exponential in terms of $m$.
In the next section, we present an equivalent LP formulation with dimension $m$ to compute $v^{*}$.


\section{Structural Properties of the Attacker's Strategy}
In this section, we investigate the structural properties of the optimal attacker's strategy. The value of the game ($v^{*}$) can be defined as follows based on the attacker's mixed strategy $p$:
\begin{eqnarray*}
v^*=\underset{p}{\text{max}}\underset{1\le i\le n_d}{\text{min}}({p}^TA)_i,
\end{eqnarray*}
where $({p}^TA)_i$ denote the $i^\text{th}$ element of ${p}^TA$. From \eqref{game matrix elements}, $({p}^TA)_i$ can be written in the following form, 
\begin{eqnarray*}
({p}^TA)_i = \sum_{j=1}^{n_a}{p_ja_{ji}}=\sum_{j=1}^{n_a}p_j\sum_{l\in{x}_j\cap{{y}_i}^c}^{}\phi_l
=\sum_{l\in {y}_i^c}^{}{\alpha_l\phi_l},
\end{eqnarray*}
where, 
\begin{equation}
\alpha_j = \sum_{\{i|j\in{x}_i\}}^{}{p_i}\implies M_{\left[ m,k_a\right]} p =\alpha \label{alpha},
\end{equation}
where $\alpha= [\alpha_1,\dots,\alpha_m]^T$, and $M_{\left[ m,k_a\right]}\in \mathbb{R}^{m\times n_a}$ is a {\it combinatorial matrix} \footnote{A combinatorial matrix $M_{[m,k]}\in\mathbb{R}^{m\times {m\choose k}}$ is a boolean matrix containing all combinations of $k$ 1's. Each column of $M$ has $k$ entries equal to 1 and rest of the entries equal to 0. In other words, $M$ is a matrix constructed from ${{m}\choose{k}}$ combinations of $k$ one in an $m$ dimensional vector}. 

Since $M_{\left[ m,k_a\right]}$ is a combinatorial matrix, $\sum_{i=1}^{m}\alpha_i = k_a$. Moreover, in the following lemma we show that for any feasible $\alpha$ there exists a feasible $p$. Hence, the problem reduces to computing $\alpha^*$.

\begin{lemma}\label{lemma5}
$M_{\left[ m,k_a\right]}$ is a surjective mapping. 
\end{lemma}
\begin{proof}
Please refer to the Appendix for the proof.
\end{proof}
Based on the above lemma, the problem reduces to computing $\alpha^*$. 

\begin{lemma} \label{Theorem2}
$\alpha^*$ satisfies the following property:
\begin{eqnarray}
\alpha^*_i\phi_i\ge \alpha^*_j\phi_j \quad\text{for}\quad i>j
\end{eqnarray}

where $\alpha_j^*$ is defined in~\eqref{alpha} for $p^*$.
\end{lemma}
\begin{proof}
Assume the following holds for $\alpha^*$:
\begin{eqnarray}
\alpha_{i_m}^*\phi_{i_m}\ge \dots \ge \alpha_{i_1}^*\phi_{i_1}.
\end{eqnarray}
Note that $v^*$ is $(m-k_d)$-sum of smallest $\alpha_l\phi_l$, that is $v^* = \sum_{l=i_1}^{i_{m-k_d}}{\alpha_l^*\phi_l} $.
Assume that there exist $i$ and $j$ such that $\alpha_i^*\phi_i< \alpha_j^*\phi_j$ for $i>j$. $ \phi_i\ge \phi_j\Rightarrow\alpha_j^*> \alpha_i^*$. $\alpha_i^*<1, \alpha_j^*>0 \Rightarrow (e_{i}-e_{j})^T\nabla_{\alpha}v|_{v^*}\le0$. Since $v^*$ is the maximum value of the $(m-k_d)$-sum of smallest $\alpha_l\phi_l$, we arrive at a contradiction. Therefore, $\alpha^*_i\phi_i\ge \alpha^*_j\phi_j$ for $i>j$. 
\qed
\end{proof}


\begin{corollary} \label{corollary2}
$\alpha^*$ and $v^*$ satisfy the following property:
\begin{enumerate} [label=(\alph*)]
\item $v^*=\sum_{l=1}^{m-k_d}\alpha^*_l\phi_l$,
\item $\alpha^*_m\phi_m=\dots=\alpha^*_s\phi_s > \alpha^*_{s-1}\phi_{s-1}\ge \dots \ge \alpha^*_{s-r}\phi_{s-r}$, $\alpha^*_{s-r-1}=\dots=\alpha^*_{1}=0$ \\for $1\le s \le \max(k_a,m-k_d)$ and $0\le r \le s-1$.
\end{enumerate}
\end{corollary}

\begin{proof}
(a) Since $v^*$ is $(m-k_d)$-sum of smallest $\alpha_l\phi_l$, This property can be concluded directly from Lemma~\ref{Theorem2}. 

(b) Let $k$ denotes $\max(k_a,m-k_d)$. First, we show that there is an optimal solution such that $\alpha^*_m\phi_m=\dots=\alpha^*_k\phi_k$. We proceed the proof by contradiction. We assume that $\exists i\in\{k+1,\dots,m\}$ such that $\alpha_i^*\phi_i>\alpha_{i-1}^*\phi_{i-1}$. Since $\sum_{l=1}^{m}\alpha_l=k_a$, there is $j\in\{1,\dots k\}$ such that $\alpha_j^*<1$. Therefore, $(e_{j}-e_{i})^T\nabla_{\alpha}v|_{v^*}\le0$. Note that if $(e_{i}-e_{j})^T\nabla_{\alpha}v|_{v^*}<0$ is a contradiction with the fact that $v^*$ is the optimal value, and if $(e_{i}-e_{j})^T\nabla_{\alpha}v|_{v^*}=0$ then it means there are multiple solutions which at least one satisfy the property. Moreover, from Lemma~\ref{Theorem2}, if $\alpha^*_m\phi_m=\alpha^*_s\phi_s$ then $\alpha^*_m\phi_m=\dots=\alpha^*_s\phi_s$, which completes the proof.
\qed
\end{proof}

Let $s^*, r^*$ denote the indices for optimal structure expressed in Corollary~\ref{corollary2} . Let ${\cal{U}}_a$ and ${\cal{U}}_d$, called active sets of attacker and defender, denote the union of $x_i$'s and $y_i$'s corresponding to the support sets of $p^*$ and $q^*$, respectively.    
\begin{corollary}\label{corollary3}
In a security game $({\cal{X}},{\cal{Y}},A)$, ${\cal{U}}_a = \{s^*-r^*,\dots,m\}$. When $s^*> m-k_d$, the defender has a pure strategy with ${\cal{U}}_d=\{m-k_d+1,\dots,m\}$, else ${\cal{U}}_d=\{s^*,\dots,m\}$ (for $s^*\le m-k_d$). 
\end{corollary}
\begin{proof}
The proof of first part directly follows from the fact that $\alpha_m,\dots,\alpha_{s^*-r^*} > 0$ and $\alpha_{s^*-r^*-1}=\dots=\alpha_1 = 0$.

For second part, consider $({p^*}^TA)_j$. The following condition holds for $U_{i^*,i^*+r^*}$: 
\begin{eqnarray*} 
&\alpha_m\phi_m=\dots=\alpha_{s^*}\phi_{s^*}>\alpha_{s^*-1}\phi_{s^*-1}\ge\dots\ge\alpha_{s^*-r^*}\phi_{s^*-r^*}\\
&\alpha_{s^*-r^*-1} =\dots=\alpha_{1}= 0. 
\end{eqnarray*}
When $k_a+k_d \le m$, $({p^*}^TA)_j > v^*$ for all $j$ such that $\{1,\dots, s-1\} \not \subseteq y_j^c$. Consequently, 
\begin{equation*}
q^*_j = 0 \quad \forall j \quad \text{s.t.} \quad \{1,\dots, s-1\} \not \subseteq y_j^c,
\end{equation*}
else ${p^*}^TAq^* > v^*$, which is a contradiction. Therefore, any $q_j^*$ corresponding to $y_j$ such that $y_j \cap \{1,\dots,s-1\}\neq \emptyset$ is zero. In other words,  ${\cal{U}}_d= \{s^*,\dots, m\}$.

Based on similar arguments, we can conclude that ${\cal{U}}_d= \{s^*,\dots, m\}$ for $k_a+k_d > m$ and $s^*\le m-k_d$. When $k_a+k_d > m$ and $s^*> m-k_d$, $({p^*}^TA)_j > v^*$ for all $j$ such that $\{s,\dots, m\} \cap y_j^c\ne \emptyset$, and consequently $q^*_j = 0$. Therefore,
\begin{equation*}
q^*_j = 0 \quad \forall j \quad \text{s.t.} \quad \{s,\dots, m\} \not \subseteq y_j
\end{equation*}
Since the defender has $k_d$ resources, it has a pure strategy to allocate it to targets $\{m-k_d+1,\dots,m\}$. 
\qed
\end{proof}
\begin{remark} 
According to Corollary \ref{corollary3}, both players choose mixed strategies that involve targets with highest impacts ($\phi_i$).
\end{remark}

\section{Computation of $v^*$}

Based on Lemma~\ref{Theorem2}, we can solve the following LP to compute $v^*$:
 \begin{equation} \label{reduced LP}
 	\begin{aligned}
 		&\underset{\alpha_1,\dots,\alpha_m}{\text{maximize}}
 		& &  \sum_{l=1}^{m-k_d}{\alpha_l\phi_l} \\
 		& \text{subject to}
 		& & \alpha_i\phi_i \ge \alpha_j\phi_j\quad \text{for all} \quad i>j\\
 		&&& \sum_{i=1}^{m}{\alpha_i}=k_a\\
 		&&& \alpha_i \le 1 \quad i =1,\dots, m.
 	\end{aligned}
 \end{equation}
Note that from Lemma~\ref{lemma5}, we show that for any $\alpha$ which satisfies the conditions in~\eqref{reduced LP}, there exists a $p$ on a simplex such that $Mp=\alpha$, which satisfy the feasibility condition of~\eqref{reduced LP}.

From Corollary~\ref{corollary2}, $v^*$ and $\alpha^*$ can be computed by examining all feasible solutions for $1\le s \le k$ ($k=\max(k_a,m-k_d)$), and $0\le r \le s-1$ which satisfy the condition in Corollary~\ref{corollary2} (b). Let $U$ denote a square matrix of dimension $k$. The $(i,i+r)^{\text{th}}$ entry of $U$ (denoted by $U_{i,i+r}$) is the solution to the following problem:  

\begin{eqnarray}\nonumber \label{off-diagonal property}
\begin{aligned}
&\underset{\alpha_1,\dots,\alpha_m}{\text{maximize}}
& & \sum_{l=1}^{m-k_d}{\alpha_l\phi_l} \\
&\text{subject to}
& & \alpha_m\phi_m=\dots=\alpha_s\phi_s>\alpha_{s-1}\phi_{s-1}\ge\dots\ge\alpha_{s-r}\phi_{s-r}\\
&&& \alpha_{s-r-1} =\dots=\alpha_{1}= 0, \qquad s=k-i+1\\
&&& 0\le \alpha_l\le 1, \forall l\in \cal{I}.
\end{aligned}
\end{eqnarray} 

The following theorem relates $v^{*}$ to the elements of $U$.

\begin{theorem} \label{Theorem3}
$v^*=\underset{i,j}{\text{max}}\{{U_{i,j}}\}$, and the entries of $U$ are as follows:

\underline{For $i\le k_a+k_d-m$:}
\begin{gather*}    
\begin{cases}
  U_{i,i}=0,\\    
  U_{i,i+r}= \sum_{l=s-r}^{m-k_d}\phi_l \quad \text{when} \quad{c_i\phi_s}\ge {k_a-r}> {c_i\phi_{s-1}}   
\end{cases}
\end{gather*}

\underline{For $i> k_a+k_d-m$:}
\begin{gather*}    
\begin{cases}
  U_{i,i}=\frac{k_a(i-k-k_d+m)}{c_i}\quad \text{when}\quad c_i\phi_s\ge k_a\\    
  U_{i,i+r}= \sum_{l=s-r}^{s-1}\phi_l+\frac{(k_a-r)(i-k-k_d+m)}{c_i}, \\ 
  \qquad \text{when} \quad c_i\phi_{s-r}> (i-k-k_d+m), \quad \text{and}\quad{c_i\phi_s}\ge {k_a-r}> {c_i\phi_{s-1}} ,\\
  U_{i,i+r} = (k_a-r-c_{i-1}\phi_s)\phi_{s-r}+\sum_{l=s-r+1}^{s-1}\phi_l+(i-k-k_d+m)\phi_s\\
  \qquad \text{when} \quad c_i\phi_{s-r}\le (i-k-k_d+m) \quad \text{and} \quad {c_i\phi_s}\ge {k_a-r}> {c_i\phi_s}-1\\
  U_{i,i+r} =0 \quad \text{otherwise}
\end{cases}
\end{gather*}
where $c_i= \sum_{j=s}^{m}{\frac{1}{\phi_j}}$.
\end{theorem}
\begin{proof}
First, we consider the case $s\leq m-k_d$. Let the optimal solution be $\alpha^*=(\alpha_1^*,\ldots,\alpha_m^*)$. Since $\alpha_i=\alpha_s(\phi_s/\phi_i)$ for $i\geq s$, $\delta\alpha_i=\delta\alpha_s(\phi_s/\phi_i)$ for any perturbation $\delta\alpha_s$. Since $\sum\alpha_i=k_a$, any allowable perturbation around $\alpha^*$ satisfies the following condition:
\begin{equation}
\sum_{j=1}^{m}\delta\alpha_j=0\implies\sum_{j=1}^{s-1}\delta\alpha_j+c_i\phi_s\delta\alpha_s=0,  
\label{perturb}
\end{equation}
where $c_i= \sum_{j=s}^{m}{\frac{1}{\phi_j}}$. Consider a perturbation that involves perturbing $\alpha_l$ for $l<s$ and $\alpha_s,\dots,\alpha_m$. From (\ref{perturb}), we obtain the following:
\begin{equation}\label{perturb1}
\delta\alpha_l=-c_i\phi_s\delta\alpha_s
\end{equation}
Based on the first order necessary conditions for maxima, we obtain the following:
\begin{equation}\label{perturb2}
\delta v|_{v^*}<0\implies\phi_l\delta\alpha_l+\sum_{j=s}^{m-k_d}\phi_s\delta\alpha_s=\phi_l\delta\alpha_l+\phi_s(m-k_d-s+1)\delta\alpha_s<0 
\end{equation}
Let $g(\phi)=-c_i\phi+(m-k_d-s+1)$. Substituting (\ref{perturb1}) in (\ref{perturb2}) leads to the condition $g(\phi_l)\delta\alpha_s<0$. $g(\phi_l)>0\Rightarrow \delta\alpha_s<0\Rightarrow\forall j<s, \alpha_j^*=1$. If for any $l<s$, $g(\phi_l)<0\Rightarrow\delta\alpha_s>0\Rightarrow\alpha_s^*=1$.

As a result, we obtain the following conditions:
\begin{eqnarray}{\label{eqn:alpha*}}
\alpha^*_j&=&1 \quad\text{if}\quad g(\phi_j)>0\\
\alpha^*_s&=&1 \quad\text{if}\quad\exists j \quad\text{such that} \quad \alpha_jg(\phi_j)<0, \nonumber
\end{eqnarray}

From (\ref{eqn:alpha*}), we conclude that $\alpha^{*}$ and $v^{*}$ can have the following forms:
\begin{enumerate}
\item 

\begin{eqnarray}\label{alphaj-offdiagonal}
\alpha_j = \left\{ {\begin{array}{*{20}{c}}
  0  && j=1,\dots,s-r-1\\ 
  1 && j=s-r,\dots,s-1\\
  \frac{k_a-r}{c_i\phi_j} && j=s,\dots,m
\end{array}} \right.  
\end{eqnarray}

From feasibility conditions in~\eqref{reduced LP} (i.e. $\sum_{l=1}^{m}\alpha_l=k_a,\alpha_j\le 1$ and $\alpha_s\phi_s>\alpha_{s-1}\phi_{s-1}$), we conclude that at $(i,i+r)^{th}$ entry of $U$, feasibility conditions are satisfied if ${c_i\phi_s}\ge {k_a-r}>{c_i\phi_{s-1}}$. Substituting~\eqref{alphaj-offdiagonal} in $U_{i,i+r}=\sum^{m-k_d}_{j=1} \alpha_j\phi_j$ leads to the following expression for $U_{i,i+r}$:
\begin{eqnarray}
U_{i,i+r}= \sum_{l=s-r}^{s-1}\phi_l+\frac{(k_a-r)(i-k-k_d+m)}{c_i} 
\end{eqnarray}

\item
\begin{eqnarray}\label{alphaj-offdiagonal-3}
\alpha_j = \left\{ {\begin{array}{*{20}{c}}
  0  && j=1,\dots,s-r-1\\ 
  \delta  && j=s-r\\
  1 && j=s-r+1,\dots,s\\
  \frac{\phi_s}{\phi_j} && j=s+1,\dots,m,
\end{array}} \right. 
\end{eqnarray}
where $\delta=(k_a-r-c_{i-1}\phi_s)$, which results from $\sum_{j=1}^{m}\alpha_j=k_a$. Since $0<\delta \le 1$, $0< k_a-r-c_{i-1}\phi_s\le 1$, which is equivalent to $c_i\phi_s\ge k_a-r>c_i\phi_s-1$.

Moreover, substituting~\eqref{alphaj-offdiagonal-3} in $U_{i,i+r}=\sum^{m-k_d}_{l=1} \alpha_l\phi_l$ leads to the following expression for $v$:
\begin{eqnarray}
U_{i,i+r}=\delta\phi_{s-r} + (i-k-k_d+m)\phi_s+ \sum_{l=s-r+1}^{s-1}\phi_l.
\end{eqnarray} 

\item

 \[\alpha_m\phi_m=\dots=\alpha_{s}\phi_{s}, \quad \alpha_s \neq 0,\quad \alpha_{j}=0\quad \text{for}\quad j\in \{1,\dots, s-1\}\] 
 Substituting the above condition in $U_{i,i+r}=\sum^{m-k_d}_{l=1} \alpha_l\phi_l$, we obtain the following:  
 \begin{eqnarray}
 &&U_{i,i+r} = \sum_{l=s}^{m-k_d}{\alpha_l\phi_l}=(i-k-k_d+m)\alpha_j\phi_j\\
 &&\implies \alpha_j= \frac{U_{i,i+r}}{(i-k-k_d+m)\phi_j},\quad j\in \{s,\dots,m\}\label{diagonalcell-alphaj}
 \end{eqnarray}
By substituting \eqref{diagonalcell-alphaj} into $U_{i,i+r}=\sum^{m-k_d}_{l=1} \alpha_lv_l$, we obtain the following:
 \begin{equation}
 \sum_{j=s}^{m}{\frac{U_{i,i+r}}{(i-k-k_d+m)\phi_j}}=k_a\implies U_{i,i+r} = \frac{k_a(i-k-k_d+m)}{\sum_{j=s}^{m}{\frac{1}{\phi_j}}}
 \end{equation}
 Let $c_i=\sum_{j=s}^{m}{\frac{1}{\phi_j}}$. Next, we have to check whether $\alpha$ satisfies the feasibility conditions of~\eqref{reduced LP}. Substituting $U_{i,i+r}$ in \eqref{diagonalcell-alphaj} leads to the following: 
 \begin{eqnarray}\nonumber
 \alpha_j = \left\{ {\begin{array}{*{20}{c}}
   \frac{k_a}{\phi_jc_i}  && j\in \{s,\dots,m\}\\ 
   0 && j\in \{1,\dots,s-1\}
 \end{array}} \right.  
 \end{eqnarray}
\end{enumerate}

Finally, we consider the case when $k_a+k_d>{m}$. Since $U_{i,i+r} = \sum_{l=s-r}^{m-k_d}\alpha_l\phi_l$ and $\alpha_{j} = 0$ for $j=1,\dots,m-k_d$, for $i=1,\dots, k_a+k_d-m$, $U_{i,i} = 0$. Moreover, $U_{i,i+r}$ can be written as $U_{i,i+r}=\sum_{l=s-r}^{m-k_d}{\alpha_l\phi_l}$, and the feasible $\alpha$'s are given as follows:
\begin{eqnarray}
\alpha_j = \left\{ {\begin{array}{*{20}{c}}
  0  && j=1,\dots,s-r-1\\ 
  1 && j=s-r,\dots,s-1\\
  \frac{k_a-r}{c_i\phi_j} && j=s,\dots,m
\end{array}} \right.  
\end{eqnarray}
If ${c_i\phi_s}\ge {k_a-r}>{c_i\phi_{s-1}}$, then $(i,i+r)^{\text{th}}$ entry of $U$ is feasible. For all $i>k_a+k_d-m$, the arguments are same as for the case $k_a \le m-k_d$.

Since $v^*$ is the maximum value which satisfies all feasibility conditions, $v^*$ is the maximum entry of $U$.
\qed

\end{proof}

Next, we show that $U$ is a sparse matrix, which leads to a linear time algorithm for computing $v^*$. Let $U^I$ and $U^{II}$ be square matrices of dimension $k$ defined as follows:

\begin{eqnarray}\label{UI condition}
U_{i,i+r}^I = \left\{ {\begin{array}{*{20}{c}}
  U_{i,i+r}  && c_i\phi_{s-r}> (i-k-k_d+m),  {c_i\phi_s}\ge {k_a-r}> {c_i\phi_{s-1}},\\ 
  0 && \text{otherwise}
\end{array}} \right.  ,
\end{eqnarray}
\begin{eqnarray}\label{UII condition-2}
U_{i,i+r}^{II} = \left\{ {\begin{array}{*{20}{c}}
  U_{i,i+r}  && c_i\phi_{s-r}\le (i-k-k_d+m), {c_i\phi_s}\ge {k_a-r}> {c_i\phi_{s}-1}\\ 
  0 && \text{otherwise}
\end{array}} \right. 
\end{eqnarray}

\begin{lemma}\label{lemma3}
Given an infeasible cell in $U^I$, either all the cells to the right (in the same row) or all the cells below (in the same column) are infeasible.
\end{lemma}
\begin{proof} Consider an infeasible cell $(i,i+r)$ in $U^I$. For a cell to be infeasible, at least one of the three inequalities in~(\ref{UI condition}) needs to be violated.
\begin{enumerate}[label=(\alph*)]
\item First, consider the case $c_i\phi_{s-1}\ge k_a-r\Rightarrow c_i\phi_{s-1}\ge k_a-r', \quad\forall r'\ge r$. In other words, if $c_i\phi_{s-1}\ge k_a-r$, there is no feasible solution in $(i,i+r')^\text{th}$ entry of $U^I$ for all $r'\ge r$. 
\item Next, consider the case, $ k_a-r > c_i\phi_{s}$. Since $c_{i+1}\phi_{s-1}=c_i\phi_{s-1}+1$, $ k_a-r+1 > c_{i+1}\phi_{s-1}\Rightarrow (i+1,i+r)^\text{th}$ entry of $U^I$ cannot be feasible. Since $i$ is arbitrary, we can conclude that $(i+j,i+r)^\text{th}$ entry of $U^I$ cannot be feasible for all $j\ge 1$. 
\item Finally, consider the case in which the inequality $c_i\phi_{s-r}> (i-k-k_d+m)$ is the only one that is violated at $(i,i+r)^\text{th}$ entry of $U^I$. Therefore, $ c_i\phi_{s}\ge k_a-r > c_i\phi_{s-1}\Rightarrow k_a-r+1 >c_{i+1}\phi_{s-1}$, and consequently, there is no feasible solution in $(i+j,i+r)^\text{th}$ entry of $U^I$ for all $j\ge 1$. Therefore, any column of $U^I$ contains at most one feasible (non-zero) entry.
\end{enumerate}
\end{proof}

\begin{corollary}\label{corr3}
At most one cell in a column of $U^I$ is feasible.
\end{corollary}
\begin{proof}
The proof follows directly from the arguments for Lemma~\ref{lemma3} (c). 
\end{proof}

\begin{theorem} \label{linear time attacker}
$v^*$, $\alpha^*$ can be computed in ${\cal{O}}(k)$ time.
\end{theorem}

\begin{proof}
   
From Lemma~\ref{lemma3} and Corollary~\ref{corr3}, we can conclude that from a current cell $(i,j)$ in $U^I$, one needs to search either in cell $(i+1,j)$ or cell $(i,j+1)$ to find the next feasible element. Therefore, a linear search (${\cal{O}}(k)$) that alternates between rows and columns leads to the cell containing the maximum element.

Next, we show that all feasible entries in $U^{II}$ can be computed in ${\cal{O}}(k)$ time. For each row $i$, there is at most one $r$ which satisfies ${c_i\phi_s}\ge {k_a-r}> {c_i\phi_{s}-1}$ in~\eqref{UII condition-2}. Therefore, for each row in $U^{II}$, we can find the feasible cell in constant time. This implies that all feasible entries in $U^{II}$ can be computed in ${\cal{O}}(k)$ time, and a linear or a logarithmic search among the feasible entries provides the maximum element. 
\qed
\end{proof}

Algorithm~1 gives $v^*,\alpha^*$  and active targets for the attacker and the defender in linear time.

\begin{algorithm}
    \caption{Computation of the value, and active targets}
    \begin{algorithmic}[1] 

    \State \textbf{Input}: $\phi_{1},\dots,\phi_m$ and $k_a, k_d$
    \State \textbf{Output}: ${v}^*,\alpha^*,{\cal{U}}_a,{\cal{U}}_d$
    \State Construct $U$ based on Theorem~\ref{Theorem3} and Theorem~\ref{linear time attacker}.

    \State $i_1 \leftarrow 1 $
    \For{\texttt{$j=1:m-k_d$}}
    	\For{\texttt{$i=i_1:j$}}
			\If {${c_i\phi_s}\ge {k_a-r}> {c_i\phi_{s-1}}$}
				\If {$c_i\phi-{s-r}>i-k-k_d+m$}
					\State $U_{i,i+r}= \sum_{l=s-r}^{s-1}\phi_l+\frac{(k_a-r)i-k-k_d+m}{c_i},$
				\EndIf
				\State $i_1 \leftarrow i$
				\State \Return $i$
			\ElsIf {${k_a-r}> {c_i\phi_s}$}
			    \State $i_1 \leftarrow i$
			    \State \Return $i$
			\Else
			\State $U_{i,i+r}=0$
			\State $i_1 \leftarrow i$
			\EndIf    
    	\EndFor
    \EndFor 

	\For{\texttt{$i=1:k$}}
		\State find $r$ such that $c_i\phi_s \ge k_a-r>c_i\phi_s-1$
    		\If {$c_i\phi_{s-r}\le i-k-k_d+m$,  and ${c_i\phi_s}\ge {k_a-r}> {c_i\phi_s}-1$}
    				    \State $U_{i,i+r} = (k_a-r-c_{i-1}\phi_s)\phi_{s-r}+\sum_{l=s-r+1}^{s-1}\phi_l+(i-k-k_d+m)\phi_s$
    				\Else
    				\State $U_{i,i+r}=0$
    				\EndIf
    	    \EndFor

    \State $v^*\leftarrow \max U_{i,j}$  
    \State $(i^*,j^*)\leftarrow \arg \max U_{i,j}$
    \State $ {\cal{U}}_a \leftarrow \{k-j^*+1,\dots,m\} $
    \State $ {\cal{U}}_d \leftarrow \{k-i^*+1,\dots,m\} $
    \end{algorithmic}
\end{algorithm}

\section{Dual Analysis: Structural Properties of the Defender's Strategy and Algorithms}
In this section, we present structural results for the optimal strategy of the defender, and present an $\mathcal{O}(m)$ algorithm to compute $v^*$ and its corresponding optimal strategy. From the definition of $v^*$, we obtain the following:
\begin{eqnarray*}
v^*=\underset{q}{\text{min}}\underset{1\le j\le n_a}{\text{max}}(Aq)_j.
\end{eqnarray*} 
where $(Aq)_j$ denote the $j^\text{th}$ element of $Aq$. From \eqref{game matrix elements}, $(Aq)_j$ can be written in the following form, 
\begin{eqnarray*}
(Aq)_j = \sum_{i=1}^{n_d}{q_ia_{ji}}=\sum_{i=1}^{n_d}q_i\sum_{l\in{x}_j\cap{{y}_i}^c}^{}\phi_l
=\sum_{l\in {x}_j}^{}{\beta_l\phi_l},
\end{eqnarray*}
where, 
\begin{eqnarray}
\beta_j = \sum_{\{i|j\in{{y}_i}^c\}}^{}{q_i} \label{beta}
\implies \beta = M_{\left[ m,m-k_d\right]}q,
\end{eqnarray}
where $\beta = [\beta_1,\dots,\beta_m]^T$, and $M_{\left[ m,(m-k_d)\right]}\in \mathbb{R}^{m\times n_d}$ is a {\it combinatorial matrix}. 
Since $M_{\left[ m,m-k_d\right]}$ is a combinatorial matrix, $\sum_{i=1}^{n_d}\beta_i = m-k_d$. Moreover, from lemma~\ref{lemma5}, for any feasible $\beta$ there exists a feasible $q$. 

The following lemma provides the structure of $\beta^*$.

\begin{lemma} \label{lemma 8}
$\beta^*$ satisfies one of the following conditions:
\begin{enumerate} [label=(\alph*)]
\item $\phi_{s-r}\ge \beta_s^*\phi_s=\dots=\beta_m^*\phi_m\ge\phi_{s-r-1}$, and $\beta_1^*=\dots=\beta_{s-1}^*=1$,
\item $\phi_{s-r}\ge \beta_s^*\phi_s=\dots=\beta_m^*\phi_m = \phi_{s-r-1}$, and $\beta_{s-1}^*\phi_{s-1}\ge\beta_s^*\phi_s$,\\ and $\beta_1^*=\dots=\beta_{s-2}^*=1$,\\
where $1\le s \le m$, $0\le r \le s-1$, $r+1\le k_a\le r+m-s$.
\end{enumerate}
\end{lemma}
\begin{proof}
Let the sequence $\{i_1,\dots,i_m\}$ of indices satisfy the following condition: 
\begin{eqnarray}
\beta_{i_1}^*\phi_{i_1}\ge \dots \ge \beta_{i_m}^*\phi_{i_m}
\end{eqnarray}
Note that $v^* = \sum_{l=i_1}^{i_{k_a}}{\beta_l^*\phi_l} $. 

First, we show that $\beta_{i_{k_a}}^*\phi_{i_{k_a}} = \beta_{i_{k_a+1}}^*\phi_{i_{k_a+1}}$. Assume 1) $\beta_{i_{k_a}}^*\phi_{i_{k_a}} > \beta_{i_{k_a+1}}^*\phi_{i_{k_a+1}}$ 2) there exists an $i\in\{i_{k_a},\dots,i_m\}$ such that $\beta_i^* <1$. Since $(e_{i}-e_{i_{k_a}})^T\nabla_{\beta}v|_{v^*}<0$ at $v^{*}$, we arrive at a contradiction. Now, assume $\beta_{i_{k_a+1}}^*=\dots=\beta_{i_m}^*=1$. Since $\sum_{l=1}^{m}{\beta_l}=m-k_d$ and $k_d\ge 1$, there exist $i,j\in\{i_1,\dots,i_{k_a}\}, i>j$ such that $\beta_i^*,\beta_j^* <1 $. Therefore, $(e_{j}-e_{i})^T\nabla_{\beta}v|_{v^*}<0$, and we arrive at a contradiction. Therefore,  $\beta_{i_{k_a}}^*\phi_{i_{k_a}} = \beta_{i_{k_a+1}}^*\phi_{i_{k_a+1}}$.  In a similar manner, we can show that $\beta_m^*\phi_m = \beta_{i_{k_a}}^*\phi_{i_{k_a}}$. 


Next, we prove that $\forall i$ such that $\beta_i^*\phi_i \neq \beta_m^*\phi_m$, there is at most one $\beta_j^* <1$ and $\beta_j^*\phi_j > \beta_m^*\phi_m$, and the rest of $\beta_i^*$'s are 1. Assume that there are $\beta_j^*<1, \beta_k^*<1$ such that $\beta_j^*\phi_j \neq \beta_m^*\phi_m$ and $\beta_k^*\phi_k \neq \beta_m^*\phi_m$. If $\beta_k^*\phi_k, \beta_j^*\phi_j > \beta_m^*\phi_m$, and $j>k $, then $(e_{k}-e_{j})^T\nabla_{\beta}v|_{v^*}<0$, and we arrive at a contradiction. Now, assume that $\beta_j^*\phi_j < \beta_m^*\phi_m$, and $\beta_j^* <1$, therefore $(e_{j}-e_{i})^T\nabla_{\beta}v|_{v^*}<0$ for all $i$ such that $\beta_i^*\phi_i > \beta_j^*\phi_j$, which leads to a contradiction. 

Next, we prove that $\beta^*$ always satisfies one of the conditions in the Lemma. Let $\Gamma = \{i|\beta_i^*\phi_i=\beta_m^*\phi_m\}$. First, we prove that $\forall j\in \Gamma$ if $\beta_j^* <1$ then $j+1 \in \Gamma$. To begin with, we assume that $\beta_j^* <1, \beta_j^*\phi_j=\beta_m^*\phi_m$ and $j+1 \notin \Gamma$. If $\beta_{j+1}^*\phi_{j+1}>\beta_j^*\phi_j$, $(e_{j}-e_{j+1})^T\nabla_{\beta}v|_{v^*}<0$, which leads to a contradiction. Moreover, if $\beta_{j+1}^*\phi_{j+1}<\beta_j^*\phi_j$ then $\beta_{j+1}^* <1$ and $(e_{j+1}-e_{i})^T\nabla_{\beta}v|_{v^*}<0$ for all $i$ such that $\beta_{j+1}^*\phi_{j+1}<\beta_i^*\phi_i$. This completes the proof for the first structure in the Lemma. Let $j= \min(\Gamma)$ and $\beta_j* =1$. Therefore, for any $i\in \Gamma$, either $\beta_i^*<1\Rightarrow i+1\in \Gamma$ or $\beta_i^*=1,\phi_j=\phi_i$. The last condition leads to the second structure in the Lemma.
\qed 
\end{proof}

Similar to the analysis for the attacker, from the above Lemma, we can compute $v^*$ and $\beta^*$ by examining all possible solutions which satisfy conditions (a) or (b) in Lemma~\ref{lemma 8}. Let $W$ be a square matrix of dimension $m$. 


\begin{theorem}
$v^*=\underset{i,j}{\text{min}}\{{W_{i,j}}\}$, where entries of $W$ are defined as follows:

\begin{gather*}    
\begin{cases}
  W_{i,i+r}=\frac{(k_a-r)(i-k_d)}{c_i}+\sum_{l=s-r}^{s-1}\phi_l,\\  
  \qquad \text{for} \quad   s-1\ge r \ge 0, r+m-s\ge k_a \ge r+1, c_i\phi_{s-r} \ge i-k_d \ge c_i\phi_{s-r-1},\\
  W_{i,i+r}=(i-k_d+1-c_i\phi_{s-r-1}) \phi_{s-1} + (k_a-r)\phi_{s-r-1}+ \sum_{l=s-r}^{s-2}\phi_l,\\
  \qquad \text{for}\quad s-1\ge r \ge 0, r+m-s\ge k_a \ge r+1, \\
  \qquad c_i\phi_{s-r-1}+1 > i-k_d +1  \ge c_{i+1}\phi_{s-r-1}\\
  W_{i,i+r}=+\infty,\quad \text{otherwise}
\end{cases}
\end{gather*}

where, $c_i = \sum_{j=s}^{m}\frac{1}{\phi_j}$.
\end{theorem}  
\begin{proof}
First case corresponds to the structure (a) in Lemma~\ref{lemma 8}. In this case, $\beta_1=\dots=\beta_{s-1}=1$. Since $\sum_{l=1}^{m}\beta_l=m-k_d$, and $\beta_s\phi_s=\dots=\beta_m\phi_m$, we obtain the following expression for $\beta_j$
\begin{equation} \label{beta1}
\beta_j = \frac{i-k_d}{c_i\phi_j}, \quad j=s,\dots,m,
\end{equation} 
where $i=m-s+1$ and $c_i = \sum_{j=s}^{m}\frac{1}{\phi_j}$. Next, we provide the feasibility conditions for structure (a). Since $0\le\beta_j\le 1$, $c_i\phi_{s} \ge i-k_d$. Moreover, $\phi_{s-r}\ge \beta_s\phi_s\Rightarrow c_i\phi_{s-r} \ge i-k_d$. Additionally, $\beta_s\phi_s\ge\phi_{s-r-1}\Rightarrow i-k_d \ge c_i\phi_{s-r-1}$. Note that $k_a$-largest terms of $\beta_l\phi_l$ contain at least one term in the set $\{\beta_s\phi_s,\dots,\beta_m\phi_m\}$. Therefore, $r+m-s\ge k_a \ge r+1$. By substituting~\eqref{beta1} into  $W_{i,i+r} = \sum_{l=i_1}^{i_{k_a}}{\beta_l\phi_l} $,
\begin{eqnarray}
W_{i,i+r}=\frac{(k_a-r)(i-k_d)}{c_i}+\sum_{l=s-r}^{s-1}\phi_l.
\end{eqnarray}


The second case corresponds to the structure (b) in Lemma~\ref{lemma 8}. In this case, $\beta_1=\dots=\beta_{s-2}=1$. Since $\sum_{l=1}^{m}\beta_l=m-k_d$ and $\beta_s\phi_s=\dots=\beta_m\phi_m=\phi_{s-r-1}$, we obtain the following:
\begin{eqnarray}\label{beta2}
\beta_j = \frac{\phi_{s-r-1}}{\phi_j}, \quad j=s,\dots,m,\\
\beta_{s-1}= i-k_d+1-c_i\phi_{s-r-1},\label{beta3}
\end{eqnarray} 
where $i=m-s+1$ and $c_i = \sum_{j=s}^{m}\frac{1}{\phi_j}$. Next, we provide the feasibility conditions for structure (b). Since $\beta_{s-1}\phi_{s-1}\ge \phi_{s-r-1}$, and $0\le\beta_{s-1}< 1$, $c_i\phi_{s-r-1}+1 > i-k_d +1  \ge c_{i+1}\phi_{s-r-1}$. Note that $k_a$-largest terms of $\beta_l\phi_l$ contain at least one term in the set $\{\beta_s\phi_s,\dots,\beta_m\phi_m\}$. Therefore, $r+m-s\ge k_a \ge r+1$. By substituting~\eqref{beta2}, \eqref{beta3} into  $W_{i,i+r} = \sum_{l=i_1}^{i_{k_a}}{\beta_l\phi_l} $,
\begin{eqnarray}
W_{i,i+r}=(i-k_d+1-c_i\phi_{s-r-1}) \phi_{s-1} + (k_a-r)\phi_{s-r-1}+ \sum_{l=s-r}^{s-2}\phi_l.
\end{eqnarray}
\qed
\end{proof}

Since $W$ is a square amtrix of dimension $m$, $v^*$ can be computed in ${\cal{O}}(m^2)$. As in the case of the defender, we can show that $W$ can be computed in ${\cal{O}}(m)$ due to sparsity of $W$ (the feasiblity consitions).   
\begin{theorem}\label{Theorem4}
$v^*$ can be computed in ${\cal{O}}(m)$.
\end{theorem}

\begin{proof}
Please refer to the Appendix for the proof.
\end{proof}

\section{Conclusion}
In this work, we address a security game as a zero-sum game in which the utility function has additive property. We analyzed the problem from attacker and defender's perspective, and we provided necessary conditions for the optimal solutions. Consequently, the structural properties of the saddle-point strategy for both players are given. Using the structural properties, we reach to the linear time algorithm, and semi-closed form solutions for computing the saddle points and value of the game.   

There are several directions of future research. One direction is to use the proposed structural properties to formulate a network design problem to minimize the impact of attacks, which leads to design resilient networks from security perspective. Another direction of future research is to generalize the results of this work to nonzero-sum games with different utility functions for attacker and defender. Finally, we plan to extend our analysis to security games with non-additive utility functions.

%
%
%
 \bibliographystyle{splncs04}
 \bibliography{ref}
\appendix
\section{Proof of Lemma~\ref{lemma5}}
\begin{proof}
The proof has already appeared in~\cite{emadi2019security}. The proof is by induction. We assume that the lemma is true for $m'=m-1$ and $k_a'=1,\dots,m-1$. $M_{\left[m,k_a \right]}$ can be written as
\begin{eqnarray}
M_{\left[m,k_a \right]}= \left[ {\begin{array}{*{20}{c}}
  \boldsymbol{1}^T&\boldsymbol{0}^T \\ 
  M_{\left[m-1,k_a-1 \right] }&M_{\left[m-1,k_a \right] }
\end{array}} \right].
\end{eqnarray} 
By separating $p$ into $\bar{p}_1$ and $\bar{p}_2$,
\begin{eqnarray} \nonumber
M_{\left[m,k_a \right]}p&=&\left[ {\begin{array}{*{20}{c}}
  \boldsymbol{1}^T&\boldsymbol{0}^T \\ 
  M_{\left[m-1,k_a-1 \right] }&M_{\left[m-1,k_a \right] }
\end{array}} \right]
\left[ {\begin{array}{*{20}{c}}
  \bar{p}_1\\ 
  \bar{p}_2
\end{array}} \right]\\
&=& \left[ {\begin{array}{*{20}{c}}
  \boldsymbol{1}^T\bar{p}_1\\ 
  M_{\left[m-1,k_a-1 \right] }\bar{p}_1+M_{\left[m-1,k_a \right]}\bar{p}_2
\end{array}} \right]. \label{feasability}
\end{eqnarray}
Second entry of the above matrix can be written in the following form:
\begin{eqnarray}
 \alpha_{1}M_{\left[m-1,k_a-1 \right] }{p'}_1+(1-\alpha_{1})M_{\left[m-1,k_a \right]}{p'}_2,
\end{eqnarray}
where, $\bar{p}_1=\alpha_{1}{p'}_1,\bar{p}_2=(1-\alpha_{1}){p'}_2$. Since we assumed that the lemma is true for $m'=m-1$ and $k_a'=1,\dots,k_a-1$, there exists ${p'}_1$ and ${p'}_2$ in a simplex such that the following hold:
\begin{eqnarray}
M_{\left[m-1,k_a-1 \right] }{p'}_1=\frac{k_a-1}{k_a-\alpha_{1}}
\left[ {\begin{array}{*{20}{c}}
  \alpha_{2}\\ 
  \vdots\\
  \alpha_{m}
\end{array}} \right],\\
M_{\left[m-1,k_a \right] }{p'}_2=\frac{k_a}{k_a-\alpha_{1}}
\left[ {\begin{array}{*{20}{c}}
  \alpha_{2}\\ 
  \vdots\\
  \alpha_{m}
\end{array}} \right].
\end{eqnarray}
By substituting the above expressions in~\eqref{feasability}, we conclude that $M_{\left[m,k_a \right]}p=\alpha$, where $p$ lies on a simplex. In other words, $\alpha$ lies in a convex-hull of columns of $M$. In order to complete the proof, we need to show that the lemma holds for the base cases $M_{\left[m+1,m \right]}$ and $M_{\left[m,1 \right]}$. 

Note that $M_{\left[m,1 \right]}=I_{m\times m}$. Therefore $p=\alpha$, and $\sum_{j=1}^{m}{p_j}=\sum_{j=1}^{m}{\alpha_{j}}=1$. Next, 
\begin{eqnarray}
M_{\left[m+1,m \right]}p=(\boldsymbol{1}\boldsymbol{1}^T-I)p=\alpha
\end{eqnarray}
\begin{eqnarray}
p&=&(\boldsymbol{1}\boldsymbol{1}^T-I)^{-1}\alpha=(\frac{\boldsymbol{1}\boldsymbol{1}^T}{m}-I)\alpha=\boldsymbol{1}-\alpha.
\end{eqnarray}
Consequently, lemma holds for base cases, which completes the proof. 
\qed
\end{proof}

\section{Proof of Theorem~\ref{Theorem4}}

\begin{proof}
Let $W^a$ and $W^b$ denote matrices of the following form: 

\begin{eqnarray}\label{Wa condition}
W_{i,i+r}^a = \left\{ {\begin{array}{*{20}{c}}
  W_{i,i+r}  && \text{satisfying structure (a) in Lemma~\ref{lemma 8}}\\ 
  0 && \text{otherwise}
\end{array}} \right.  ,
\end{eqnarray}
\begin{eqnarray}\label{Wb condition}
W_{i,i+r}^b = \left\{ {\begin{array}{*{20}{c}}
  W_{i,i+r}  && \text{satisfying structure (b) in Lemma~\ref{lemma 8}}\\ 
  0 && \text{otherwise}
\end{array}} \right.  .
\end{eqnarray}

First, note that all feasible entries of $W^a$ and feasible entries of $W^b$ are disjoint due to complimentary feasibility conditions ($ i-k_d \ge c_i\phi_{s-r-1}$ in $W^a$, and $ i-k_d < c_i\phi_{s-r-1}$ in $W^b$). Moreover, since $c_i\phi_{s-r} \ge i-k_d \ge c_i\phi_{s-r-1}$, for any $s$ there is at most one specific $r$ which satisfies conditions of $W^a$. This implies that computation of all feasible entries of $W^a$ is in ${\cal{O}}(m)$. Therefore, any row of $W$ has at most one feasible entry of $W^a$.

Next, we show that computing all feasible entries of $W^b$ is in ${\cal{O}}(m)$. From the second structure of Lemma~\ref{lemma 8}, for any $\hat{s},\hat{r}$ and $\hat{i}=m-\hat{s}+1$, if $\beta_{\hat{s}-1}>1$ ($\beta_{\hat{s}-1}<0$), ${(i,\hat{i}+r)^\text{th}}$ entry of $W$ is infeasible for all $i>\hat{i},r >\hat{r}$ ($i<\hat{i},r <\hat{r}$) since it implies $\beta_{s-1}>1$ ($\beta_{s-1}<0$). 
Therefore, at every entry of $W$, $\beta_{s-1}$ provides a criteria for which the rest of entries, in the same row and in the same column are entirely infeasible, and consequently it is not required to check the feasibility of those entries. In other words, value of $\beta_{s-1}$ provides a criteria for direction of searching for feasible entries of $W^{b}$. Thus, computing feasible entries of $W^{b}$ is in ${\cal{O}}(m)$. 
\qed
\end{proof}
\end{document}